\documentclass[sn-nature]{sn-jnl}

\usepackage{graphicx}%
\usepackage{multirow}%
\usepackage{amsmath,amssymb,amsfonts}%
\usepackage{mathrsfs}%
\usepackage[title]{appendix}%
\usepackage{xcolor}%
\usepackage{textcomp}%
\usepackage{manyfoot}%
\usepackage{booktabs}%
\usepackage{listings}%

\usepackage{amsthm}

\usepackage{csquotes}
\usepackage{amsfonts}

\makeatletter
\DeclareRobustCommand*\cal{\@fontswitch\relax\mathcal}
\makeatother

\usepackage[ruled, noend]{algorithm2e} 
\usepackage{algpseudocode}

\usepackage{pgfplots}
\usepackage{tikz, tikz-3dplot}
\usetikzlibrary{arrows,shapes,automata,calc,matrix,backgrounds,petri, positioning,pgfplots.statistics}
\pgfplotsset{compat=1.10}

\makeatletter
\DeclareRobustCommand*\cal{\@fontswitch\relax\mathcal}
\makeatother

\newtheorem{definition}{Definition}
\newtheorem{remark}{Remark}

\newtheorem{theorem}{Theorem}

\newtheorem{lemma}{Lemma}

\begin{document}

\title[Market Equilibria With Buying Rights]{Market Equilibria With Buying Rights}


\author*[1]{\fnm{Martin} \sur{Loebl}}\email{loebl@kam.mff.cuni.cz}


\author[2]{\fnm{Anetta} \sur{Jedli\v{c}kov\'{a}}}\email{anetta.jedlickova@fhs.cuni.cz}

\author[3]{\fnm{Jakub} \sur{\v{C}ern\'{y}}}\email{cerny@disroot.org}


\affil[1]{\orgdiv{Department of Applied Mathematics, Faculty of Mathematics and Physics}, \orgname{Charles University}, \orgaddress{\street{Malostranské nám. 25}, \city{Prague}, \postcode{11800},  \country{CZE}}}
 
\affil[2]{\orgdiv{Department of Philosophy, Faculty of Humanities}, \orgname{Charles University}, \orgaddress{\street{Pátkova 2137/5}, \city{Prague}, \postcode{18200}, \country{CZE}}}

\affil[3]{\orgdiv{Department of Industrial Engineering and Operations Research}, \orgname{Columbia University}, \orgaddress{\street{500 W. 120th Street}, \city{New York}, \postcode{10027}, \state{NY}, \country{USA}}}






\abstract{
We embed buying rights into a (repeated) Arrow–Debreu model to study the long-term effects of regulation through buying rights on arising inequality. Our motivation stems from situations that typically call for regulatory interventions, such as rationing, namely, distribution crises in which demand and supply are persistently misaligned. In such settings, scarce resources tend to become increasingly concentrated among more affluent individuals, while the needs of the broader population remain unmet. While fully centralized distribution may be logistically or politically unfeasible, issuing buying rights offers a more practical alternative: they can be implemented digitally, e. g.,  via tokens traded on online platforms, making them significantly easier to administer. We model a scenario in which a regulator periodically distributes buying rights with the aim of promoting a more equitable allocation. Our contributions include (i) the definition of the (iterated) market where in each round the buying rights are distributed and then traded alongside the resource, (ii) the approximation algorithm of the market-clearing prices in every round, and (iii) the upper bound on \textit{frustration} -- a notion conceptually similar to the Price of Anarchy, but for systems regulated through buying rights, defined as the arising loss in fairness the individual buyers have to take when the distribution is handled via the market.
}

\keywords{Arrow–Debreu model, multi-round trading environment, rationing, fairness}

\maketitle

\section{Introduction}
\label{sec.intro}
Crises of distribution arise when a vital resource becomes severely limited. Such crises can occur for various reasons, such as natural disasters, wars, or economic instabilities. Different mechanisms can be deployed to distribute resources in these situations. Two approaches may be considered extremes: a centralized distribution by delegated authorities and a fully decentralized distribution through intervention-free markets. Markets have the potential to distribute goods flexibly and reliably among many buyers and sellers~\cite{cripps2006efficiency}. However, the free market has its drawbacks. In the absence of sufficient supply, scarcity can lead to significant price increases, with some traders attempting to acquire more than their \enquote{fair} share at the expense of others. This situation benefits the most powerful or well-connected individuals or organizations, leaving the less fortunate with limited or no access to vital resources. On the other hand, centralized distribution by authorities can ensure a fairer division of resources based on pre-established rules. However, resources are often not owned by the central authority, and, in addition, this approach comes with economic and time-related inefficiencies~\cite{moroney1997relative}.

We focus specifically on the stable period of most severe distribution crises, when there is a significant mismatch between the claim and a steady, albeit small, resupply of goods and money over an extended period of time.  The setting also requires to have a benchmark ``fair'' centralized allocation, given the needs of the buyers and the supply of the suppliers. 
We further assume that the distribution of the resource according to the benchmark allocation is not directly achievable by the central authority, but the authority can implement periodic buying-right distribution policies. The aim of the authority is to move the distribution of the good indirectly towards the benchmark allocation. These features of a distribution crisis are realistic; for example, they have similarities to the COVID-19 pandemic crisis that motivated this investigation.

To combine the advantages of both the free market and centralized approaches, Martin Loebl formulated the iterative hybrid market mechanism with buying rights (see \cite{JLS}). In~\cite{JLS}, the authors define the iterative distribution system, where each iteration consists of two stages. The system, called {\em Crisis}, has been further developed and analyzed in \cite{CFLS} using numerical methods. In each iteration, the central authority first distributes buying rights to the available amount of critical commodity among the buyers in a desirable manner. This first stage exemplifies a \enquote{claim problem} as studied in the literature on axiomatic resource allocation~\cite{T}. Unlike traditional interpretations where conflicting claims arise from legal debt, in our case, claims represent buyers' \enquote{needs} for the critical resource. The authority divides the amount of the resource into buying rights of the same total quantity, rather than directly into portions of the commodity, and the system moves to the second stage.

The second stage is a market where two types of goods are traded through money: the critical commodity, only sold by sellers, and the buying rights, only sold by buyers. This stage can be implemented using various market models, but a crucial restriction must hold: {\em At the end of each market, each buyer must own at least as many rights as the amount of the commodity they possess.}
We also consider a more restricted setting, called {\em Restricted Crisis} in which {\em the money obtained by the sales of rights can only be used to buy the commodity, and only in later iterations.} The restricted setting adds an advantage to active buyers, who can buy more commodity earlier than remaining buyers. 
All rights are eliminated at the end of each market and the commodity acquired by buyers is consumed according to their needs.

To accentuate the desired centralized distribution perspective prioritized during crises, we refer to the division rules used in the first stage of the extended market to distribute buying rights as the {\em (rights) distribution mechanisms}. In crises, these rules are often based on societal and ethical preferences and aim to divide the rights more equitably. Our goal is to minimize the role of the central authority while increasing equity compared to unregulated markets. 
The aim is to broaden the distribution of the essential commodity to buyers to whom the central authority intends to allocate goods, but is unable to do so because of their limited financial resources. This is achieved by enabling buyers to accumulate funds through the sale of their buying rights during trading rounds. At the same time, more active buyers need to pay extra for additional rights. Hence, the proposed multi-round trading mechanism can be viewed as an 'autonomous taxation system'. 

The restricted setting in which money obtained from selling rights can be used only in the subsequent markets gives to the active buyers the advantage, important in a crisis, to obtain the resource earlier than passive buyers. The proposed system is also beneficial for sellers since it broadens the trading.

\subsection{Preliminaries}
Iterative market mechanism Crisis ($\mathbb{C}$) consists of a sequence of Markets, $\mathbb{C}=(\mathbb{M}^{\tau})_{\tau \leq T}$.There is one scarce resource called {\it Good}. The two other commodities, representing funds and the right to purchase Good, will be referred to as {\it Money}\footnote{For a similar treatment of Money as a commodity in redistribution markets, see \cite{DKA}.}
and {\it Right}. In this paper, the commodities of Good and Right are {\em indivisible}, while Money is a {\em divisible} commodity. In this trading environment, traders form two disjoint sets of {\it sellers} and {\it buyers}. Sellers engage in selling Good, while buyers engage in selling Right, buying Right, or purchasing Good. 

Each single-round Market of Crisis $\mathbb{C}$ has two stages. The first stage begins with the sellers declaring the amount of Good for sale and the buyers declaring their {\it Claims} for Good. The total Claim may exceed the total offer.  The authority distributes Right of the same total quantity as the total declared amount of Good among the buyers.
At the beginning of each Market, each seller receives the declared amount of Good and each buyer receives an endowment of Money. 

In the second stage called {\em Trading}, the traders exchange their \emph{initial endowments} of Good, Right, and Money. Sellers are restricted to selling Good, while buyers are allowed to buy Good and Right and sell Right but cannot sell Good. A key restriction is: \emph{at the end of each Market, each buyer must have at least as much Right as Good}. In other words, a unit of Right is required to legally acquire a unit of Good. 

In the {\em Restricted Crisis}, it is required in addition that {\em Money obtained by the sales of Right can only be used to buy Good, and only in later Markets.}

At the end of each Market, (1) buyers consume all the obtained Right, (2) buyers consume all the obtained Good and keep the remaining Money for the next Market, and (3) sellers only sell Good and consume the obtained Money. Hence, the Markets in the sequence forming a Crisis are interdependent. 

In this paper, we however consider the case of {\em myopic buyers} who optimize each single-round Market of the Crisis $\mathbb{C}$. 
Markets of the Crisis with myopic buyers can be considered independently, and we devote next two sections to the study of the single-round Market. 
The independence of single-round Markets is not valid for Restricted Crises with myopic buyers, by the nature of the additional restriction. 

To evaluate the effectiveness of this iterative system, we use the same measure as in~\cite{JLS} and \cite{CFLS} called \enquote{frustration}. Frustration is defined as the normalized difference between the amount of Right assigned to a buyer and the actual amount of Good purchased by the buyer (normalized by the  amount of Right assigned to a buyer). Hence, frustration is by definition at most 1.  
This measure captures the degree to which the allocation of the benchmark authority is violated by the outcome of the distribution. Without introducing buying rights, frustration may reach 1. 

\subsection{Main Contributions}
In this work, we present and analyze the multiround trading environment (Restricted) Crisis, introduced in \cite{JLS}, to distribute a scarce good from myopic sellers to myopic buyers, who optimize each round, with heterogeneous wealth levels during crises. 
Our results are derived with two assumptions. These assumptions are realistic since we focus specifically on the `stable' period of most severe distribution crises. The initial and terminal periods of distribution crises are our future work.
{\em First,} we study the mechanism under the assumption that agents are myopic: they wish to maximize their utility in every individual market. This is a realistic assumption, since in the middle stable part of a crisis, strategic planning is hard.  
{\em Second,} we assume that buyers (for example hospitals) have practically unlimited amounts of money in each individual market, and use their individual sub-additive utility function to decide how to spend them. This setting came up as more practical than an upper bound on the amount of money buyers have, in the COVID project mentioned earlier, which started this research. In fact, the analysis is easier in the latter case. 

Under these assumptions, we show in Section \ref{sec.alg} that there is an efficient algorithm which approximates an optimal solution of each Market (see {\bf Theorem \ref{thm.mmain}}). The algorithm is analogous to the standard approximation algorithm for the Arrow-Debreu model described in Chapter 5 of \cite{nisan2007algorithmic}). Its analysis is, however, novel and more complicated in several ways, in particular considering both additive and sub-additive utilities.

In the last section, we study the effect of Crisis on the frustration of the buyers. Our second positive result (see {\bf Theorem \ref{thm.pot}}) demonstrates that when Right is used, {\em potential frustration} in each Market of the Crisis is always capped at 1/2. This is optimum since we also provide examples of Markets where the potential frustration of one buyer is equal to 1/2. 
We further analyze the Restricted Crisis and show ({\bf Theorem~\ref{thm.wp}}), under further assumptions, that frustration is also capped at 1/2. The (potential) frustration of 1/2 compares favorably with the frustration of the free market, which can be as high as 1. 

\subsection{Acknowledgment} This project receives funding from the Horizon EU Framework Programme under Grant Agreement No. 101183743. We would like to thank to David Sychrovský for many helpful discussions.

\section{Related Work}

The allocation of resources to individuals in a desirable -- especially fair -- manner has been extensively studied over the past decades. The objective of fair distribution is to identify an allocation mechanism that satisfies certain properties, commonly known as fairness criteria. In the literature, there is a wide variety of notions of fairness, and numerous works have explored the possibility of achieving both fairness and efficiency simultaneously. These works have examined fairness notions such as Envy Freeness, in relation with Pareto optimality, and maximum Nash welfare. A survey in \cite{far} provides an overview. Our mechanism also draws from the theory of claims and taxation problems, particularly fair divisions in bankruptcy problems, as surveyed in \cite{T}.

Our work contributes to the field of redistributive mechanisms, with a specific focus on reducing inequalities. In the literature, a related study examines a two-sided market for trading goods of homogeneous quality, aiming to optimize the total utility of traders \cite{DKA}. However, our approach differs as we consider more general fairness measures beyond social welfare, and we assume that utilities are common knowledge. This line of research has been expanded to include settings with heterogeneous quality of tradable objects, various measures of allocation optimality, and imperfect observations of traders \cite{akbarpour2020redistributive}. Another related work explores multiple market and non-market mechanisms for allocating a limited number of identical goods to multiple buyers \cite{condorelli2013market}. The author argues that when buyers' willingness to pay aligns with the designer's allocation preferences, market mechanisms are optimal, and vice versa. In crises, where critical resources are highly valuable to all participants but some lack the necessary funds to acquire them, it is in society's best interest to allocate goods fairly. These findings suggest that relying solely on unregulated markets may not be the best approach during a distribution crisis.

Emissions allowances and tradable allowance markets share similarities with our work. Historically, regulators allocated tradable property rights directly to firms, leading to inefficiencies such as misallocation, regulatory distortions, and barriers to entry. Contemporary market designs utilize auctions for the allocation of tradable property rights. Tradable allowance markets, as discussed in \cite{dor}, play a crucial role in ensuring the efficiency of carbon markets and preventing market power exertion by large and dominant agents. Efficiency of multi-round trading auctions for the allocation of carbon emission rights is studied in \cite{WCZ}.

A general setting of Crisis with non-myopic traders is studied using numerical methods in ~\cite{CFLS}. Crisis with non-myopic traders is studied by analytic methods in ~\cite{SLC}. 

\section{Market Model}
\label{sec: trading env}

\newcommand{\market}{\mathbb{M}(S, B, G, V, M, D, R, u, \phi)}
\newcommand{\tmarket}{\mathbb{M^{\tau}}(S^{\tau}, B^{\tau}, G^{\tau}, V^{\tau}, M^{\tau}, D^{\tau}, R^{\tau}, u^{\tau}, \phi^{\tau})}

We fix a set of sellers $S$ and a set of buyers $B$. We denote $T = S \cup B$ the set of all traders and assume $S \cap B = \emptyset$. We will assume that Good and Right are discrete entities consisting of {\em indivisible items} while Money consists of {\em divisible items}.

{\bf Market} $\mathbb{M}$ is defined as tuple $\mathbb{M}= \market$, where 
$G = (G_s|s\in S)$ and $G_s$ denote the set of items of Good seller $s$ has at the beginning of the Market. Similarly, $M = (M_b|b\in B)$ and $M_b$ denote the amount of money each buyer $b$ has at the beginning of the Market. We let $V= \sum_{s\in S} G_s$ be the offered volume of Good in the Market $\mathbb{M}$. We also denote subsets of a set with a subscript, for example $G_A = (G_a|a\in A)$ for a set $A\subset T$. The Claim $D = (D_b|b\in B)$ gives the amount of Good each buyer hopes to acquire in the Market. Further, $R = (R_b|b\in B)$ and $R_b$ is the set of items of Right distributed to the buyer $b$ by the central authority
 (rights) distribution mechanism $\phi= (\phi_b)_{b\in B}$ of Market $\mathbb{M}$: for $b\in B$, $|R_b|= \phi_b(V, D)$. 
 \noindent The rights distribution mechanism of Market $\mathbb{M}$, modeled as $\phi: \mathbb{N}^{|B| + 1} \to \mathbb{N}^{|B|}$, distributes Right for all offered Good according to individual Claims. 
Our results do not depend on a particular choice of a distribution mechanism. During a crisis, which can take various forms, the authorities may adopt different rights distributions by optimizing different functions based on the crisis's specific features. 

\noindent Finally, the utility $u= (u_t|t\in T)$ of trader $t$, having amount $x$ of Good and amount $y$ of Money, is  $u_t^G(x)+ u_t^M(y)$  where $u_t^G, u_t^M$ are functions defined as follows.

\begin{definition}[Utility of traders in $\mathbb{M}$]
\label{def.utility}
 The {\rm Good-utility function} of trader $t$ in market $\mathbb{M}$ is  a $u_t^G:\mathbb{R}\rightarrow \mathbb{R}$ where $u_t^G(x)$ denotes the $t$'s utility of amount $x$ of Good. 
 
 The {\rm Money-utility function} of trader $t$ in market $\mathbb{M}$ is  $u_t^M:\mathbb{R}\rightarrow \mathbb{R}$ where $u_t^M(y)$ denotes the $t$'s utility of amount $y$ of Money.

We require that both $u_t^G, u_t^M$ are {\rm non-negative} and {\rm monotone}. Moreover, $u_t^M$ is linear, $u_t^G$ is subadditive and sellers have a positive utility only for Money. 
\end{definition}


{\bf Rights Assignment} of $|R_b|= \phi_b(V, D)$ items of Right to each buyer $b$ by the distribution mechanism $\phi$ constitutes the {\bf first step} of a Market. In the {\bf second step}, traders trade assigned Good, Right and Money in the {\bf Trading} which is a standard market with a single restriction that the final basket of each buyer has the amount of Right at least as big as the amount of Good. We also consider {\bf Restricted Trading} where we also require that the Money obtained for selling Right in $\mathbb{M}$ cannot be used to buy Good in $\mathbb{M}$. The buyers can use the obtained Money only in the following Markets of the Crisis.
If $X$ be a set of Good, Right and Money, we use ${\cal M}(X)$ (${\cal G}(X), {\cal R}(X)$ respectively) to denote the amount of Money (Good, Right respectively) in $X$.

\begin{definition} [Solution of $\mathbb{M}$]
\label{def.sol} 
A {\em solution} of Market $\mathbb{M}$ consists of (1) the price $q$ ($p$ respectively) per item of Right (Good respectively); the price of a unit of Money is fixed to $1$, and (2) a partition of a subset of $\cup_{t\in T} (G_t \cup R_t \cup M_t)$  into {\em baskets} $B_t$, $t\in T$, which satisfies the following condition (*): the total price of  $B_t$ is at most the total price of $t$'s initial endowment, and further if $t$ is a buyer then ${\cal R}(B_t)\geq {\cal G}(B_t)$.
A solution is {\bf optimal} if for each trader $t$, $u_t^\tau({\cal G}(B_t^\tau), {\cal M}(B_t^\tau))$ is maximum among all $u_t^\tau({\cal G}(Q), {\cal M}(Q))$, where $Q$ is a subset of the union of all the initial endowments, which satisfies the condition (*) above. Finally, solution is {\bf $t$-restricted} if $B_t$ contains the amount of Money which $t$ obtained from selling Right, i.e., 
$M_t\geq p{\cal G}(B_t)$. 
It is {\bf restricted} if for each trader $t$, $B_t$ is $t$-restricted.
\end{definition}

\section{Approximation Algorithm for Solving the Market}
\label{sec.alg}

Next, we introduce a polynomial auction-based algorithm for finding an approximate optimal (restricted optimal respectively) solution of Market $\mathbb{M}$. The algorithm is analogous to the standard approximation algorithm for the Arrow-Debreu model of \cite{nisan2007algorithmic}), with three comodities: divisible Money with linear traders utilities, indivisible Good with sub-additive traders utilities, and indivisible Right with zero traders utilities. 
There are several issues in our setting (sub-additive utility function for Good, the restriction on the amount of Right and Good in the final baskets, and one additional condition in the restrictive setting) that make its analysis more complicated. 

We recall that $M_t$ denotes the initial endowment of Money of trader $t$ and for ease of notation, we denote by $R_t$ the initial endowment of Right of $t$. 

The algorithm, divided into iterations, auctions items of Good and Right. We will denote the current price of one item of Good (Right, respectively) by $p$ ($q$ respectively). We assume that the price of one unit of Money is equal to $1$ during the execution of the algorithm. 

Initially, buyers buy items of Good.  It follows from assumptions of Theorem \ref{thm.mmain} and the description of the algorithm (Lemma 2 of the proof of Theorem \ref{thm.mmain}) that all items of Good are bought for the initial price in the first iteration of the algorithm. This also assures that the initial prices are below the equilibrium prices. 

In further iterations, one item of Good is always bought along with one item of Right.
For the purpose of clarity of the algorithm description, we introduce a new commodity called {\em Couple} as a pair $(z_1,z_2)$ where $z_1$ is an item of Good and $z_2$ is an item of Right. Let ${\cal C}(X)$ denote the number of items of Couple in set $X$. For a trader $t$, we will assume their {\em utility for $x$ items of Couple} is the same as for $x$ items of Good, i.e., equal to $u_t^G(x)$. We will denote the current price of one item of Couple by $c$. In the first iteration, buyers form items of Couple from the items of Good they buy and items of Right they have or buy. As explained above, in the setting of Theorem \ref{thm.mmain}, after the first iteration only items of Couple are auctioned. 

The algorithm finds traders’ output baskets and establishes the price. Only the final output baskets need to adhere to the rules of the Market.  When {\em approximating a restricted optimal solution}, we keep invariant during algorithm execution that the cash (an auxiliary commodity defined in the algorithm description) of each buyer b is always at least
 $q(|R_b|-|C(B_b)|)$ where $B_b$ is the current basket of $b$\footnote{The invariant ensures that cash obtained by selling Right is not used to buy Good.}.

\medskip\noindent
{\bf The algorithm description} is analogous to the standard approximation algorithm for the Arrow-Debreu model of \cite{nisan2007algorithmic}) with three comodities: divisible Money with linear traders utilities and fixed unit price, indivisible Good with sub-additive traders utilities, and indivisible Right with zero traders utilities. 

Let $0< \epsilon< 1$. Initially, we let $p= q\gets 1$. Each buyer gets the surplus {\em cash} covering its initial endowment of Money and Right, as in the standard approximation algorithm for the Arrow-Debreu model of \cite{nisan2007algorithmic}. Cash is a dummy commodity representing the flow of Money in the system.

The algorithm is divided into {\em iterations}. In each iteration, buyers form items of couple from items of Good and Right they currently have. We always have $c=p+q$ and in particular, initially $c=2$. 

In the beginning of each iteration, each item of Couple is offered for $c$ and analogously for Good and Right. However, during the iteration, some items of Couple are offered for $c$ and some for $(1+\epsilon)c$. 
An iteration ends when the price of {\em each item} of Couple is raised from $c$ to  $(1+\epsilon)c$.  

Each iteration is divided into {\em rounds}.

 {\bf Round:} we fix an arbitrary order of buyers and consider them one by one in this order. Let buyer $b$ be considered. Let us denote by $o^b$ the number of items of Couple $b$ currently has, and by $o_+^b$ the number of items of Couple $b$ currently has of price $(1+\epsilon)c$. 

Let $I^b$ ($I$ as ideal) be a set of items of Couple and of Money of max total utility which $b$ can buy with $co^b$ plus its current cash (respecting the invariant of the algorithm if we aim for the restricted optimal solution). 

If ${\cal C}(I^b)< o^b$ then $b$ does nothing, the algorithm moves to the next buyer\footnote{If this occurs, the current basket of $b$ is optimal for the previous price $c / (1+\epsilon)$ and $o_+^b= 0$.}.

If ${\cal C}(I^b)\geq  o^b$ then $b$ buys items of Couple via the Outbid.

{\bf Outbid}:
\begin{itemize}
\item
The system buys with cash one by one and at most ${\cal C}(I^b)-  o_+^b$ items of Couple for price $c$ and sells them to $b$ for cash price $(1+\epsilon)c$ per item, maintaining the invariant. First, it buys from $b$ itself.  
\item
When the system purchases a Couple $(z_1,z_2)$ for buyer $b$, the $z_1$ and $z_2$ may be bought separately from different traders, and composed into the item of Couple. This happens when some items of Right and (necessarily the same amount of items of) Good are not yet coupled in previous rounds. We observe that this happens only if they are available for the initial price from the traders. In this situation, the system again buys items of Right first from the buyer $b$.
However, the system pays nothing if it buys items of Right from an initial endowment of a buyer for the initial price since the payment is already in the surplus cash. {\bf End of Outbid}
\end{itemize}

If no more Couple is available at price $c$  during or after the Outbid then the current round and iteration terminate, $p\gets (1+\epsilon)p$ , $q\gets (1+\epsilon)q$ , 
$c\gets (1+\epsilon)c$ and the cash is updated:  
everybody who had Good or Right in its initial endowment gets extra cash, $\epsilon p$ per item of Good or $\epsilon q$ per item of Right. {\bf End of Round}

If a round went through all buyers, the algorithm proceeds with the next round. {\bf End of Iteration}

When nobody wants to buy new items of Couple, the whole trading ends. The system takes all Money from the buyers, sells them to both the buyers and the sellers for cash and keeps whatever quantity remains.
The output of the algorithm consists of (1) the collection of the final baskets of each trader and (2) the terminal prices  $p, q, c$.
{\bf This finishes the description of the algorithm.}

\subsection{Analysis of the Algorithm}
\label{sub.anal}
Next, we show the correctness of the algorithm and upper bound its complexity. Although the algorithm is analogous to the standard approximation algorithm for the Arrow-Debreu model of \cite{nisan2007algorithmic}, the analysis is more involved since the nature of Market is different than the nature of the trading environment considered in \cite{nisan2007algorithmic}. In the analysis, we use assumptions of two kinds. First, we make the natural algorithmic assumption that {\em each run of the Outbid, which is carried out by the central authority, takes one algorithmic step}.  Secondly, we assume that {\em the sizes of the initial endowments of the buyers satisfy the three properties stated in Definition \ref{def.fe} below}. Properties (1),(2) are simple to achieve by rescaling the units of Money. Property (3) follows from the assumption mentioned in the Introduction, that buyers (for example hospitals) have practically unlimited amounts of money and use their individual sub-additive utility function to decide how to spend them.

\medskip

\begin{definition}
\label{def.fe}
The initial endowments $(R_b, M_b); b \in B$ are {\em valid} if they satisfy, for each buyer $b\in B$, the following properties: 
  (1) $M_b> 4R_b$, 
(2) for each $x\leq R_b$, $u_b^G(x)\geq  2u_b^M(x)$ and
(3) for each $x\geq 1/2M_b$, $u_b^M(x)> u_b^G(x)$ and
for each $x\geq D_b$, $u_b^G(D_b)= u_b^G(x)$.
\end{definition}

\medskip

\begin{theorem}
 \label{thm.mmain}
  Let $0< \epsilon< 1$. We assume that initial endowments of buyers are valid. We further assume that each run of the Outbid takes one algorithmic step. Let $m$ denote the total initial endowment of all buyers. The following holds.
  \begin{enumerate}
 \item
  The time-complexity of the auction-based algorithm is at most 
  $|B|^2(\log_2 |V|)(1+\log_{1+\epsilon}m)$; hence, it is polynomial in the input size. 
 \item
 For each participant, its basket assigned by the algorithm is feasible and its price plus $1$ is bigger than the total price of its initial endowment. 
 \item
 The terminal price of Right is equal to the terminal price of Good.
 \item
 Relative to terminating prices, each buyer or seller gets a basket of utility at least $(1-\epsilon)$ times the utility of its optimal feasible basket.
 \end{enumerate}
 \end{theorem}
\begin{proof}

We denote by $m, r, g$ the total initial endowments of Money, Right and Good of all the traders and recall that $r= g$ since the offered volume of Good satisfies $V= \sum_{s\in S} G_s$ (see Section 3).
We prove the theorem in a sequence of lemmas.

\begin{lemma}
\label{l.B1}
 At each stage of the algorithm, the total amount of cash among the buyers is at most $2m$. We recall cash is a dummy commodity introduced in this section which represents the flow of Money in the system.
 \end{lemma}
 \begin{proof}
 Lemma is true in the beginning by assumption (1) of Definition \ref{def.fe}. The total amount of cash is gradually decreasing during each iteration, since the system always buys for less than it sells. At the end of each iteration, the system gives back to the buyers the amount of cash it earned during that iteration.
 \end{proof}
 
\begin{lemma}
\label{l.B2}
 In the first iteration, all items of Good and Right are paired.
 \end{lemma}
 \begin{proof}
 By assumptions (2) of Definition \ref{def.fe}, all buyers prefer to buy at least the fair amount of items of Couple for the initial price; by assumption (1) there is enough cash in the initial surplus of each buyer to do it.
\end{proof}
 
\begin{lemma}
\label{lem: end of iteration}
 After the end of the first iteration: (1) a buyer owes to the system only cash for its initial endowment of money and (2) the total cash among all traders is always at most $m$.
 \end{lemma}
 \begin{proof}
 The first part follows from Lemma \ref{l.B2} since all items of Good and Right are sold and bought at the end of the first iteration. For the second part, we note that among sellers, the total cash is $pg$ since all items of Good were sold in the first iteration and among buyers, the total cash is at most $m-2pg+pg$ since the buyers paid for the items of Couple (for items of Good to the sellers and for items of Right to the buyers), and they kept cash for the initial endowments of Right since all items of Right were sold and bought in the first iteration.
 \end{proof}

\begin{lemma}
\label{lam: bound on number of rounds}
  The number of rounds in each iteration is at most $2+ |B|$.
\end{lemma}
\begin{proof}
We observe that in each fully completed round, either none of the buyers buy items of Couple and the trading ends, or none of the buyers buy items of Couple in the next round and the trading ends, or at least one buyer act for the last time in this iteration. If in the current round every buyer buys items of Couple only from itself, then in the next round nobody buys since no one got additional cash. Hence, let a buyer $b$ buy items of Couple from another buyer in the current round. This means that $b$ has no Couple items for $c$, otherwise it would have to buy these first by the rules of the outbid. Hence, in this round nobody buys from $b$ meaning $b$ receives no additional cash in this iteration and the current round is the last active one for $b$.
\end{proof}

\begin{lemma}
\label{lam: bound on complexity of rounds}
  Each buyer must perform at most $\log_2 |V|$ algorithmic steps in each round.
\end{lemma}
\begin{proof}
All utility functions are known and monotone, $u_b^M(x)$ is linear and $u_b^G(x)$ is sub-additive. Let $X$ be the current basket of buyer $b$.  In the algorithm, buyer $b$ does not need to know $I^b$, it only needs to find maximum $k$ such that $|{\cal G}(X)|+k\leq D_b$ and $u_b^G({\cal G}(X)+k)- u_b^G({\cal G}(X)) > u_b^M(ck)$. This can be done by the binary search, given the assumptions on utilities.
 
\end{proof}
 
\begin{lemma}
\label{lem: bound on number of iterations}
 The total number of iterations is at most $1+\log_{1+\epsilon}m$.
 \end{lemma}
  \begin{proof}
 Each iteration raises the price of Couple by the factor of  $(1+\epsilon)$ and the max price per unit of Couple cannot be bigger than the total surplus. This argument is essentially identical to the proof of Lemma 5.24 for the auction algorithm of Chapter 4.2 of \cite{nisan2007algorithmic}.
 \end{proof}
 
\begin{lemma}
\label{lem: bound on utility}
 Relative to terminating prices, each buyer or seller gets a basket of utility at least $(1-\epsilon)$ times the utility of its (restricted) optimal basket.
 \end{lemma}
  \begin{proof}
  (1) Buyers owe nothing to the system since after the end of the trading they keep only the items of Money they can buy with their remaining cash.
 
 (2) After the end of the trading and buying items of Money, each participant is left with the amount of cash less than $1$ by the description of the last stage of the algorithm.
 
(3) Baskets of sellers are optimal (see Definition \ref{def.sol} ) since all items of Good were sold by Lemma \ref{l.B2}. 

(4) The only reason why the basket of a buyer $b$ is not (restricted) optimal is:
For some items of Couple, $b$  paid $(1+\epsilon)c$ where $c$ is the terminal price of Couple. 
Let $x$ ($y$, respectively) denote the total number of items in the Couple (Money, respectively) in a (restricted) optimal basket of $b$, whose utility is thus $u_b^G(x)+ u_b^M(y)$. This basket is market clearing and thus $y+cx\geq M_b$. 

 We observe that $y> cx$ since otherwise $cx\geq M_b/2$ (using the inequality $y + cx \geq M_b$) and by Definition \ref{def.fe} (3),
$u_b^G(x)+ u_b^M(y)\leq u_b^G(cx)+ u_b^M(y)< u_b^M(cx+y)$ which contradicts the optimality of $x,y$.
 
In $b$'s {\em terminal basket},  there are $x$ Couple items and at least $y- \epsilon cx$ amount of Money. 
The utility of $b$'s {\em terminal basket} is, using $y> cx$ and the linearity of the utility of Money, at least 
$u_b^G(x)+ u_b^M((1-\epsilon)y)=  u_b^G(x)+ (1-\epsilon) u_b^M(y)$. 

\end{proof}

Now we are ready to prove Theorem \ref{thm.mmain}. Let us address each point separately
 
(1)  It follows from Lemmas~\ref{lam: bound on number of rounds}, \ref{lam: bound on complexity of rounds}, \ref{lem: bound on number of iterations} that the time complexity of the auction-based algorithm behaves asymptotically as 
$|B|^2(\log_2 |V|)(1+\log_{1+\epsilon}m)$.

(2) follows from (1) of Lemma~\ref{lem: end of iteration}, 

(3) follows from the description of the algorithm and 

(4) is Lemma \ref{lem: bound on utility}.
\end{proof}

\section{Effect of Rationing on the Frustration of Buyers}
\label{s.wp}
When the Good is traded, the final allocation may differ from the ideal centralized distribution. This disparity serves as a measure of the inefficiency inherent in trading when it comes to allocating the Good in a manner that aligns with the central authority's preferences.

\subsection{Frustration}

The concept of the Right can be understood as the socially determined entitlement of a buyer to a specific amount of Good. We recall from the Introduction that the {\em frustration} of a buyer is the normalized difference between the amount of Right he is assigned and the amount of Good he acquired in a Market if that is at least zero, and zero otherwise. 
 This concept resembles the price of anarchy, quantifying the cost incurred by the system due to the selfish behavior of the involved actors~\cite{koutsoupias1999worst}.

\begin{definition}
\label{def.fru}
The {\em frustration}, denoted by $f_b$, of buyer $b\in B$ who acquired ${\cal G}_b={\cal G}(B_b)$ items of Good in market $\mathbb{M}$, is 
\begin{equation*}
    f_b = 
    \max\left\{0, \frac{\phi_b(V, D)- {\cal G}_b}{\phi_b(V, D)}\right\}.
\end{equation*}
\end{definition}

\begin{remark}[Potential frustration in a Market is at most $1/2$]
\label{rem.realf}
In our setting, the notion of frustration has two caveats. First, in a Market of a Restricted Crisis, buyers cannot use the acquired Money to buy Good, and their frustration is thus the same as in the free market. Second, possibly the subadditive utility function does not allow buyers to use the acquired Money to buy Good; {\em their frustration is by choice}. These caveats are the reason to consider {\em potential frustration}, denoted by $pf_b$, of buyer $b\in B$ in market $\mathbb{M}$, as
\begin{equation*}
    pf_b = 
    \max\left\{0, \frac{\phi_b(V, D)- {\cal G'}_b- d'/(d+d')|\phi_b(V, D)-{\cal G'}_b|}{\phi_b(V, D)}\right\},
\end{equation*}
where ${\cal G'}_b$ is the amount of Good $b$ acquired using his initial Money endowment $M_b$, $d$ is the price $b$ payed for one item of Good and $d'$ is the price $b$ payed for one item of Right. Hence, $pf_b\leq 1-[d'/(d+d')]$. 
\end{remark}
By Theorem \ref{thm.mmain}, we know that our auction-based algorithm gives  a (restricted) solution where all items of Right are sold, $d= d'$ and it approximates the efficient solution of the Market within an arbitrary constant. We conclude that
\begin{theorem}
    \label{thm.pot}
Potential frustration of each efficient solution of the Market is at most 1/2, and it is equal to 1/2 for each buyer who bought no Good for his initial endowment of Money. 
\end{theorem}

Since myopic traders optimize each Market of a Crisis, we conclude that in a Crisis with myopic traders, potential frustration of a buyer is at most $1/2$. 

\medskip

Next, we study the frustration in Restricted Crisis consisting of a sequence of Markets $(\tmarket)_{\tau \leq T}$, when all sellers and buyers are myopic. We are only able to make observations on the developing frustration under further assumptions.

\begin{definition}
We say a restricted Crisis is measurable if
\begin{enumerate}
    \item The total supply and the individual Claims do not change in the sequence $(\mathbb{M}^{\tau})_{\tau \leq T}$. {\em This is quite natural in the middle of a distribution crisis.}  
   
    \item The individual utility for Good satisfies: let us denote by $m^{\tau}_b$ the amount of Money $b$ is willing to spend in the Market $\mathbb{M}^{\tau}$ and by $z^\tau_b$ the amount of Money $b$ obtained by selling Right in $\mathbb{M}^\tau$. For each $\tau< T$,
$$
m^{\tau+1}_b= m^{\tau}_b + z^\tau_b.
$$
 {\em Hence, the subadditive utility function of buyer $b$ allows him to spend in $\mathbb{M}^{\tau+1}$ the same amount of Money as in the previous Market $\mathbb{M}^{\tau}$, plus the amount of Money he acquired by selling Rights in $\mathbb{M}^{\tau}$. } 
\end{enumerate}
\end{definition}

We further observe: if the sellers are myopic then the offered volume of Good satisfies for each $\tau$ that $V^\tau= \sum_{s\in S} G_s^\tau$.

The measurability assumptions are quite restrictive but arguably natural in our regime of myopic traders in a stable middle part of a distribution crisis. Note that if buying rights are not traded, the frustration can easily rise to 1. With Right, however, we can again prove an upper bound of $1/2$ to the frustration of each buyer.

\begin{theorem}
\label{thm.wp}
In all but possibly the first Market of a measurable restricted Crisis, if the efficient solution is approximated by the auction-based algorithm of section \ref{sec.alg}, each individual frustration is at most $1/2$. 
\end{theorem}

\begin{proof}We say that a buyer is frustrated in a specific Market if their frustration after the Market is non-zero. Let Market $\mathbb{M}^\tau, \tau\geq 1$ of the Crisis end and let us consider the next Market $\mathbb{M}^{\tau+1}$. By the assumptions of the theorem, the auction-based algorithm repeats the steps of Market $\mathbb{M}^\tau$.  After the final step of the auction for $\mathbb{M}^\tau$, the willingness to pay of non-frustrated (in $\mathbb{M}^\tau$) buyers is saturated.  However, frustrated (in $\mathbb{M}^\tau$) buyers continue buying Couple 
since they acquired additional funds in $\mathbb{M}^\tau$. Let $b$ be such a frustrated buyer. Let $b$ sold, in Market $\mathbb{M}^\tau$, $n_b$ items of the Right for the total price $z_b$. Hence, $b$ is willing to buy an additional  number of items of Couple for $z_b$ units of Money. Let $S_b$ be the set of $n_b$ items of Couple containing the items of Right buyer $b$ sold so far in $\mathbb{M}^{\tau+1}$. Buyer $b$ buys the Couple of $S_b$ at an increased price which in turn frees funds of active buyers who may buy back. 

\begin{lemma}
\label{l.7}
Frustration of each buyer is a non-increasing function of the market number. Especially, the non-frustrated (in $\mathbb{M}^\tau$) buyers remain non-frustrated.
\end{lemma} 
\begin{proof}The assertion holds since each frustrated buyer $b$ can buy at most $|S_b|$ additional items of Couple even if the price is not increased from the final price of $\mathbb{M}^\tau$. Hence, $b$ buys Couple only from $S_b$. The non-frustrated buyers may buy back Couple from sets $S_b$ but not beyond these sets since their willingness to pay is the same as in $\mathbb{M}^\tau$: if Couple of all sets $S_b$ is bought back by non-frustrated buyers, then the distribution of Couple is the same as at the end of $\mathbb{M}^\tau$ but the price is higher and thus no non-frustrated buyer is willing to buy more Couple.
\end{proof}

Hence, we only need to rule out the case that the final frustration of each frustrated buyer $b$ is strictly bigger than $1/2$. Let $b$ be a buyer with frustration bigger than $1/2$ in $\mathbb{M}^\tau$. Let $0< n'_b< n_b$ be such that $n_b-n'_b= R_b/2$ where we denote by $R_b$ the number of assigned rights to $b$ in $\mathbb{M}^\tau$ (and thus also in $\mathbb{M}^{\tau+1})$. 

\begin{lemma}When the price reaches double the price of Couple in $\mathbb{M}^\tau$, each frustrated buyer $b$ (1) gains $n'_b$ additional items of Couple and thus has frustration $1/2$, (2) spends all $z_b$ additional items of Money, (3) keeps $R_bz_b/n_b\geq z_b$ items of Money for sold Rights and (4) the non-frustrated buyers are not willing to buy back any of these items of Couple acquired by the frustrated buyers.
\end{lemma}
\begin{proof}
When buying $n'_b$ additional items of Couple of $S$, buyer $b$ only needs to buy items of Good by which $b$ spends all $z_b$ additional items of Money it got from the previous Market $\mathbb{M}^\tau$: \begin{itemize}
\item
$2n'_bz_b/n_b$ items of Money for buying $n'_b$ items of Good from $S_b$, and 
\item
$(n_b-2n'_b)z_b/n_b$ items of Money needed to increase the price of $R_b- n_b= 2(n_b-n'_b)-n_b$ items of Good $b$ already has, since $n_b- n'_b= R_b/2$.
\end{itemize}

(3) follows since $b$ sold $R_b/2$ items of Right and the price is doubled. Hence it remains to show (4). Non-frustrated (in $\mathbb{M}^\tau$) buyers from which $b$ bought new items of Couple are not willing to buy back since: They obtain in total $2n'_bz_b/n_b$ items of Money for the sold $n'_b$ items of Couple of $S_b$, but in order to increase the price further, $2(n_b-n'_b)z_b/n_b$ items of this obtained Money is needed to increase the price of the remaining $(n_b-n'_b)$ items of Couple in $S_b$. Clearly by the definition of $n'_b$ and since $n_b\leq R_b$, $(n_b-n'_b)= R_b/2\geq n'_b$. \end{proof}

Summarizing: the final price of Couple in $\mathbb{M}^{\tau+1}$ is at most double of the price of Couple in $\mathbb{M}^\tau$ and a frustrated buyer has frustration at most $1/2$ in $\mathbb{M}^{\tau+1}$, while the non-frustrated buyers remain non-frustrated in $\mathbb{M}^{\tau+1}$.
\end{proof}

\section{Conclusion} 

We present a multi-round trading environment that combines free market and centralized distribution. The goal is that the equilibrium of the free market is moved towards the centralized solution, which is desirable in times of need. The system assigns the buyers rights, modeled as an additional commodity, representing the amount of goods they are entitled to according to the centralized solution. The traders trade both rights and goods. To evaluate the effectiveness of the redistribution, we use the concept of \enquote{frustration}, which measures the gap between what a buyer obtained and what he was entitled to. We present a polynomial algorithm that approximates the efficient solution of each round of this trading environment. The algorithm is analogous to the Arrow-Debreu model of \cite{nisan2007algorithmic}. 
However, there are several issues in our setting (subadditive utility function, the restriction on the amount of Right and Good in the final baskets) that make the algorithm's analysis more complicated. We study this multi-round environment analytically and show that in the case of myopic buyers following the auction-based algorithm described in the previous section, the frustration of each buyer in the multi-round system is upper bounded by 1/2.

\bibliography{references}

\end{document}